\documentclass[12pt]{article}

\usepackage{fullpage}

\RequirePackage{amsthm,amsmath,amsfonts,amssymb}
\RequirePackage{natbib}
\RequirePackage{graphicx}

\usepackage[normalem]{ulem} 
\RequirePackage[OT1]{fontenc}
\usepackage{latexsym}
\usepackage{etoolbox}
\usepackage{calc}
\usepackage{fancyhdr}
\usepackage{ifthen}
\usepackage{tabularx}
\RequirePackage[colorlinks,citecolor=blue,urlcolor=blue]{hyperref}
\usepackage{epsfig,xcolor}
\usepackage{comment}
\usepackage{bm}
\newcommand{\mI}{\ensuremath{\mathbb I}}
\newcommand{\mE}{\ensuremath{\mathbb E}}
\newcommand{\mP}{\ensuremath{\mathbb P}}
\newcommand{\mR}{\ensuremath{\mathbb R}}

\newtheorem{remark}{Remark}[section]

\numberwithin{equation}{section}
\newtheorem{prop}{Proposition}[section]
\newtheorem{theorem}{Theorem}[section]
\newtheorem{lemma}{Lemma}[section]

\newcommand{\bn}{{\bf n}}
\newcommand{\bq}{{\bf q}}
\newcommand{\bp}{{\bf p}}
\newcommand{\bone}{{\bf 1}}
\newcommand{\bD}{{\bf D}}
\newcommand{\bh}{{\bf h}}
\newcommand{\bX}{{\bf X}}
\newcommand{\btheta}{{\bm \theta}}

\begin{document}

\noindent {\sffamily\bfseries\Large  Optimal multiple testing and design in clinical trials}

\vspace{3mm}

\noindent%
\textsf{Ruth Heller}, Department of
Statistics and Operations Research, Tel-Aviv university, Tel-Aviv  6997801, 
Israel,  \textsf{E-mail:} ruheller@gmail.com\\
\textsf{Abba Krieger},
Department of statistics, University of Pennsylvania, Philadelphia, USA, \textsf{E-mail:}
krieger@wharton.upenn.edu.\\
\textsf{Saharon Rosset}, Department of
Statistics and Operations Research, Tel-Aviv university, Tel-Aviv  6997801, 
Israel,  \textsf{E-mail:} saharon@tauex.tau.ac.il

\begin{abstract}
A central goal in designing clinical trials is to find the test that maximizes power (or equivalently minimizes required sample size) for finding a false null hypothesis subject to the constraint of type I error. When there is more than one test, such as in clinical trials with multiple endpoints, the issues of optimal design and optimal procedures become more complex. In this paper we address the question of how such optimal tests should be defined and how they can be found. We review different notions of power and how they relate to study goals, and also consider the requirements of type I error control and the nature of the procedures. This leads us to 
an explicit optimization problem with objective and constraints which describe its specific desiderata.  We present a complete solution for deriving optimal procedures for two hypotheses, which have desired monotonicity properties, and are computationally simple. For some of the optimization formulations this yields optimal procedures that are identical to existing procedures, such as Hommel's procedure or the procedure of \cite{Bittman09}, while for other cases it yields completely novel and more powerful procedures than existing ones. We demonstrate the nature of our novel procedures and their improved power extensively in a simulation and on the APEX study \citep{Cohen16-short}. 
\end{abstract}

\section{Introduction}

In a typical clinical trial setting the researcher is first required to determine the sample size. This calculation balances the desired power if the null hypothesis is false, with a prespecified requirement for Type I error control when the null hypothesis holds. 
The analysis is regulated: the hypotheses and sample size choices are made prior to data collection.

The three elements of the problem are 1) the objective: maximize power (at a specified material effect as postulated by the alternative hypothesis); 2) the condition: subject to  a prespecified Type I error under the null hypothesis; 3) the decision rule: values of the test statistic that lead to retaining or rejecting the null hypothesis. 

The above discussion is appropriate in the setting where there is a single hypothesis test under consideration. Often however, that is not the case. 
Multiple confirmatory endpoints are increasingly common in Phase III clinical trials  \citep{Dmitrienko18}, in addition to being  (almost) always considered in Phase I and II clinical trials.

The conceptual issues, as motivated by the three elements, are more complex when more than one null hypothesis  is tested. 
One common solution for Type I error control is through {\em strong Family-Wise Error Rate (FWER) control}  by insisting that 
 the probability of rejecting any of the true null hypotheses does not exceed $\alpha$ for all possible parameter values.
No standard formulation exists for  the power objective: 
its  nature  depends on the criterion that the researcher wants to optimize, and it typically requires making assumptions about the actual 
parameter values.

The most common approach in clinical trials with multiple endpoints is to choose an off-the-shelf multiple testing procedure (MTP) for the analysis. \cite{Ristl19} provide a comprehensive review of approaches  for  analyzing multiple endpoints in various clinical trial settings. For example,  when testing both primary and secondary endpoints, fixed sequence and hierarchical procedures are considered in order to test secondary endpoints only if the primary endpoints were rejected \citep{Dmitrienko18}. As another example, when considering contrasts of means, or the difference of means of various treatments from a common control, 
the test statistics have a specific dependence structure and the MTP takes 
the joint distribution of the test statistics into account \citep{Bretz10}.  

For the chosen off-the-shelf procedure, in the design of a clinical trial, the only remaining challenge is to determine the necessary sample size. 
In the  setting of comparing multiple treatments with a control, optimal sample size allocation has been addressed for single step MTPs  in \cite{Horn98}, for step-down MTPs in \cite{Hayter91}, and for  step-up MTPs in \cite{Dunnett01,Wang16}.

A second approach is to find the procedure that maximizes a desired aspect of power  within a selected set of allowed procedures, rather than start from a selected MTP. For example,  within the single step weighted Bonferroni procedures, the problem of optimizing the weights was considered by \cite{spjotvoll72,Westfall98,  Dobriban15}. 
As another example, \cite{Lehmann05} derive optimal policies under the severe restriction that the procedure has to be monotone in the following strong sense: if the value of the rejected $p$-value is decreased, and the value of the non-rejected $p$-value is increased, the set of rejections remains unchanged. This restriction is much stronger than our definition of weak monotonicity presented below, which allows rejection of  a $p$-value close to $\alpha$ if the other $p$-values are fairly small, but not reject it otherwise. 

A third approach starts from a sensible local test for intersection hypotheses and uses the closed testing procedure of \cite{Marcus76}. Closed testing procedures necessarily provide strong FWER control, and optimal multiple testing procedures are necessarily closed testing procedures, see \cite{Goeman21} and the references within.  Since closed testing procedures do not necessarily result in rejections of any single hypothesis, \cite{Bittman09} provided a modification that removes from the rejection region all realizations that lead to non-consonant decisions (where a consonant decision is one that rejects at least one  individual null hypothesis, \citealt{Romano11}) and adds to the rejection region realizations that lead to consonant decisions. 
They used Stouffer's local test \citep{stouffer1949american}.
 For the case of two endpoints, a
comprehensive comparison  of the suggestion in  \cite{Bittman09} with the Bonferroni and Simes local tests is given by \cite{Su12}.

\cite{Rosenblum14, Rosset18} considered  finding the optimal procedure for a power objective of interest while controlling in the strong sense a desired error rate, without restrictions to a selected set of allowed procedures. Determining the optimal multiple testing (OMT) procedure is computationally very challenging, so solutions were provided  only for $K=2$ in \cite{Rosenblum14} and for $K\leq 3$ in \cite{Rosset18}. We find in this work that imposing restrictions of interest may not only result in a more attractive procedure, but also one that is computationally far simpler.

 In \S~\ref{sec-elements} we discuss  the  elements of the multiple testing problem, and
we formulate our goal as an optimization problem.   
In \S~\ref{sec-K2} we provide our main methodological contribution for $K=2$ hypotheses. Some of the optimization problem formulations we tackle lead to OMT solutions that correspond to well known MTPs. Other formulations lead to new MTPs. In \S~\ref{sec-numerical} we compare numerically the MTPs. An interesting perspective we explore is determining optimal sample sizes and sample splitting strategies for such multiple testing scenarios. 
For example, given a predetermined total sample size, when is it optimal to split the sample size equally across the two null hypotheses?  Is it best to test both hypotheses with limited power, or to test only one with greater power? The answers turn out to be non-trivial.  Next, we provide a motivating example following \cite{Dmitrienko18}, and we revisit it in more detail in 
\S~\ref{sec-APEX}. 
In \S~\ref{sec-discuss} we conclude.

\subsection{Motivating example: a clinical trial with two populations}\label{subsec-APEX}

Following \cite{Dmitrienko18}, we use the APEX (Acute medically ill venous thromboembolism prevention with extended duration Betrixaban) trial to discuss analysis approaches in multiple population trials. The trial's goal was to examine the advantage of Betrixaban over Enoxaparin in patients at risk of venous thrombosis. For this purpose, patients who were hospitalized for acute medical conditions were randomly assigned to the two treatment arms. Patients with an elevated D-dimer level were expected to be at greater risk of venous thromboembolism and have greater benefit from the new therapy. 

The original trial design \citep{Cohen16-short} tested hypotheses in a prespecified order, so that testing stops at the first hypothesis with $p$-value greater than $\alpha$, see details in  Appendix A. 
\cite{Dmitrienko18} showed that with the Hochberg multiple testing procedure \citep{Hochberg88}, the superiority of Betrixaban would have been declared in the all-comers population. They argue that procedures that allow an efficacy finding in one cohort even though efficacy is not established in another cohort may be preferred, thus arguing in favor of a design using Hochberg's procedure over testing in order.

We argue that the choice of procedure need not be based on one off-the-shelf procedure or another, but rather on the  power objective that the researchers have in mind. To demonstrate our suggested approach, we shall use the APEX trial data as well. 
We  consider 
the following two cohorts: the subgroup with elevated D-dimer level (3870 patients, 1914 received Betrixapan); and all-comers without elevated D-dimer level (2416 patients, 1198 received Betrixapan).  Since these are two disjoint populations, one could argue in favour of testing each  population at the nominal (two sided 5\%) level. But we analyze the two cohorts together (as in \citealt{Dmitrienko18}), with the aim of providing strong FWER control at the nominal (two sided 5\%) level.  

Table \ref{tab-data}
 shows the results in the two subgroups we are examining: 
\begin{table}
\caption{\label{tab-data} no.events/no.arm for the control arm (Enoxaparin) and for the treated arm (Betrixapan) for two subgroups in the APEX trial. The $p$-value is the probability of observing at most the difference observed in the sample proportions between Enoxaparin and Betrixaban (using the normal approximation to the two independent binomial proportions) assuming the null hypothesis that the risk of venous thrombosis with Betrixaban is the same within each arm.} 
\begin{center}
\begin{tabular}{lccc}
  \hline
  
& Enoxaparin & Betrixaban & $p$-value \\ 
\hline
  Group 1: with elevated D-dimer level  & 166/1956 & 132/1914 & 0.032\\ 
  Group 2: All-comers excluding group 1  & 57/1218 & 33/1198 & 0.006  \\ 
   \hline
\end{tabular}
\end{center}
\end{table}

In order to find the optimal multiple testing (OMT) procedure for a particular power objective, we need to have a specific alternative in mind. Our power objective for each hypothesis is computed assuming that the event rate is 7.5\% in the control group, and there is an expected relative reduction of 35\% in the treatment group (so the rate is $0.65\times 7.5\% = 4.875\%$ in the treatment group). These were the parameters originally used in the APEX trial for determining the necessary sample sizes \citep{Cohen16-short}. 

We concentrate on procedures that offer strong FWER control at level $\alpha=0.025$ (to account for testing one-sided hypotheses). 
Figure \ref{fig-rejectionregion} shows 
two popular off-the-shelf  procedures: the closed-Stouffer \citep{Henning15}, which  applies the closure principle to Stouffer's test for the intersection hypothesis;  and the Hommel procedure  \citep{Hommel88},  which applies the closure principle \citep{Marcus76} to the Simes test \citep{Simes86}  for the intersection hypothesis. We chose these two procedures since they use two different types of intersection tests: Stouffer is an additive combination type intersection test (which has good power when all null hypotheses are false); Simes relies on one of the ordered $p$-values that is at most $\alpha$ (thus it has better power when only few of the null hypotheses are false). 

\begin{figure}
 \begin{tabular}{cc}
\includegraphics[width=5.5cm,height=5.5cm]{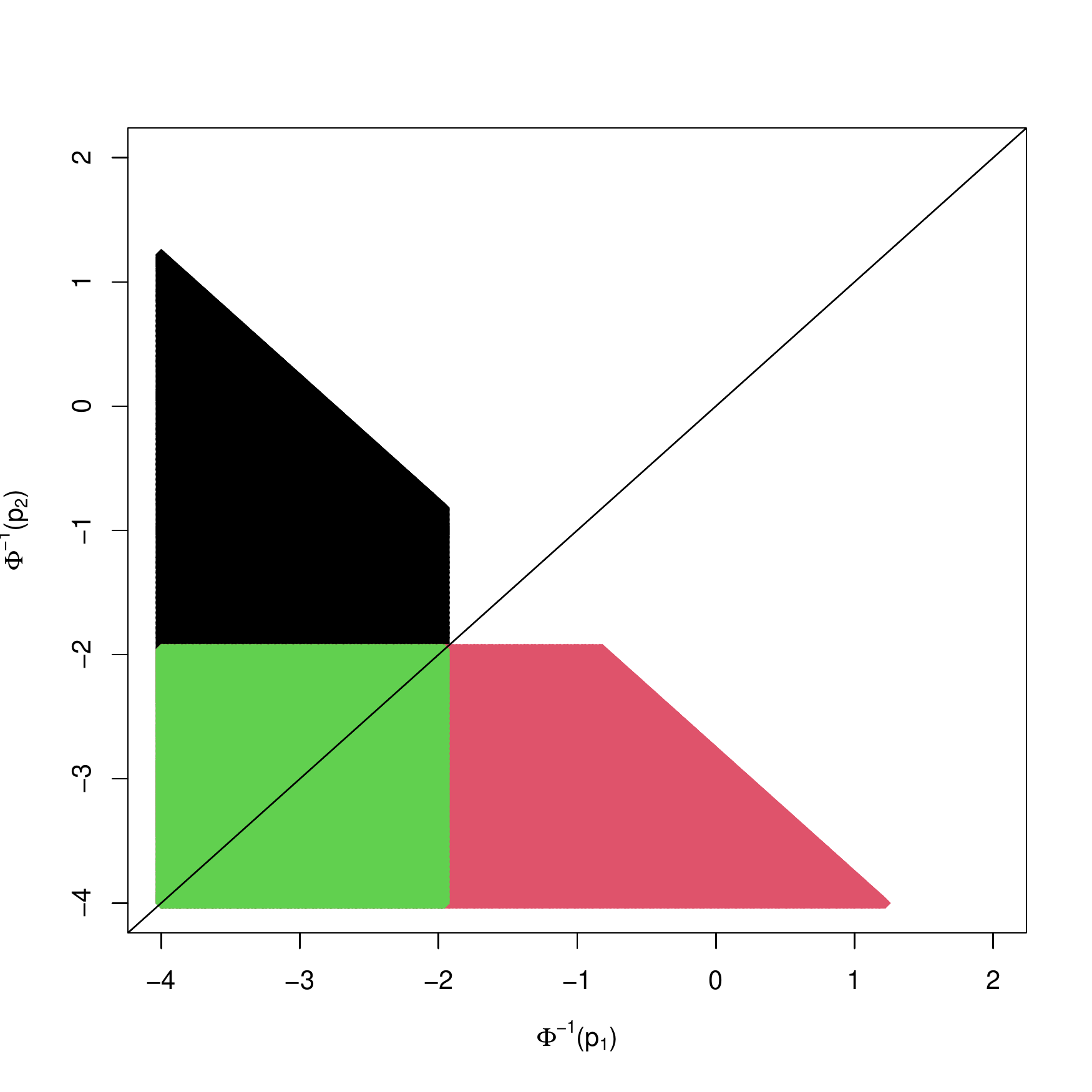}&
  \includegraphics[width=5.5cm,height=5.5cm]{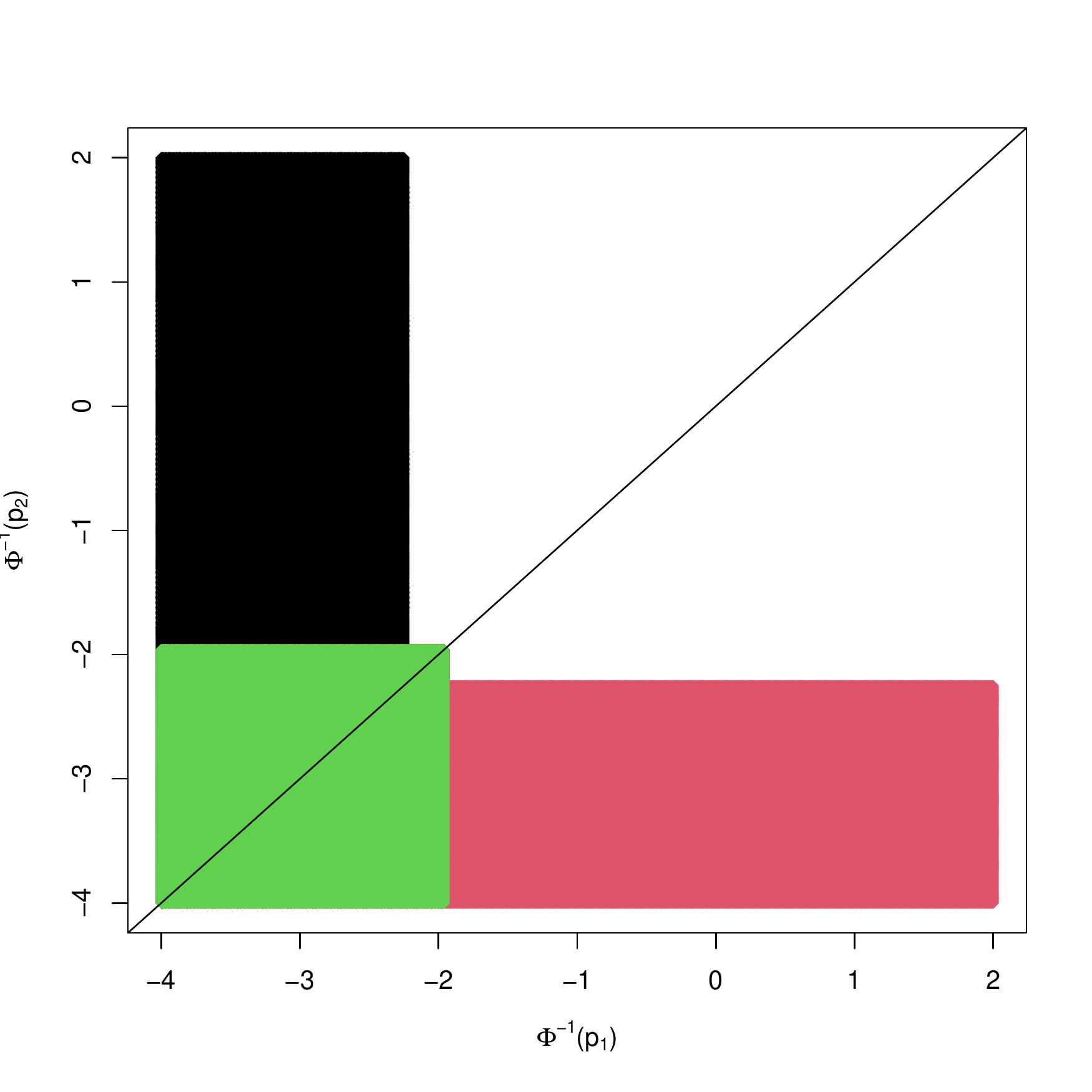}       
\end{tabular}
 \caption{\label{fig-rejectionregion}    T
 he off the shelf procedures closed-Stouffer (left) and  Hommel (right).    
   In green: reject both hypotheses; in red: reject only the second hypothesis; in black: reject only the first hypothesis. The regions are depicted for the $z$-scores, i.e., for $(\Phi^{-1}(p_1), \Phi^{-1}(p_2))$.
  }
\end{figure}

In \S~\ref{sec-APEX} we compare and contrast these procedures with novel ones, and we show  that the obtained power (objective) of novel procedures guided by optimality considerations can outperform the off-the-shelf procedures. 

\section{The elements of the problem for multiple hypotheses}\label{sec-elements}

We define a statistical model that is very general and encompasses many designs that are  encountered in practice. Assume there is a parameter space $\Theta$. The data, $X$, are generated by $\mP_\vartheta$ for $\vartheta \in \Theta$. 
We consider testing $K$ null hypotheses, and denote by $\theta_k(\vartheta)$ the parameter of interest in the $k$th hypothesis testing problem, $\theta_k: \Theta\longrightarrow \Theta_k\subseteq \mR, k=1,\ldots,K$. The binary hypothesis state vector  is denoted by $\bh(\vartheta) = (h_1(\vartheta),\ldots, h_K(\vartheta))$, so $h_k(\vartheta)=0$ if $\theta_k(\vartheta)$ is in the null subset of $\Theta_k$.  For example, we may consider $h_k(\vartheta)=0$ if $\theta_k(\vartheta)\geq 0$ and $h_k(\vartheta)=1$ if $\theta_k(\vartheta)<0$, $k=1,\ldots, K$.  
As another example, in the normal means problem,  $\vartheta = ({\bm \mu}, {\bm \Sigma}), $ and  designs of interest are as follows:  multiple treatments with a control, where $\theta_k(\vartheta) = \mu_k-\mu_{K+1}, \ k=1,\ldots,K$; pairwise comparisons, where each $\theta_{k}(\vartheta)$ corresponds to a $\mu_i-\mu_j$. 

The problem the researcher faces is to find the mapping from the sample space to the decision space, $\bD(\bX)= (D_1(\bX), \ldots, D_K(\bX))$, where $D_k(\bX) =1$ if the decision is to reject the $k$th null hypothesis, and $D_k(\bX) =0$ otherwise. 

Henceforth we suppress the dependence of $\bh$ and ${\btheta} = (\theta_1,\ldots, \theta_K)$ on $\vartheta$, 
and of $\bD$ on $\bX$. 

We want to make as many correct rejections of null hypotheses, henceforth also referred to as true discoveries. An incorrect rejection, i.e., a rejection of a true null hypothesis, is referred to as a false discovery.

\subsection{The power objective }\label{subsec-objective}  
For measures of power, we consider the probability of at least one true discovery  as well as the total number of true discoveries: 
$$
\Pi_{any}(\bD)= \mP \left(\sum_{k=1}^Kh_kD_k > 0\right); \quad \Pi_{total}(\bD) =\mE \left(\sum_{k=1}^Kh_kD_k\right).
$$
These objectives can take on different meanings and interpretations, depending on what we assume about the parameters $\vartheta,$ which affect the distribution and the value of the indicators $h_k.$

For a fixed set of alternatives, i.e., $\vartheta \in \Theta$  such that $\bh(\vartheta) = \bone$,
$\Pi_{any}(\bD)=\mP_\vartheta \left(\sum_{k=1}^K D_k > 0\right)$ and $\Pi_{total}(\bD)= \mE_\vartheta \left(\sum_{k=1}^K D_k\right)$ are the probability of making any discoveries (considered in \citealt{Bittman09, Romano11}) and the expected number of discoveries, respectively, if all null hypotheses  are false, and have the specific parameters implied by $\vartheta$.

More generally, if $\vartheta$ has prior distribution $\Lambda$, then $\Pi_{total}(\bD) = \int(\mE(\sum_{k=1}^Kh_kD_k\mid \vartheta)d\Lambda(\vartheta)$ and   $\Pi_{any}(\bD) = \int\mP\left(\sum_{k=1}^Kh_kD_k > 0\mid \vartheta\right)d\Lambda(\vartheta)$ (considered in \citealt{Rosenblum14} for $K=2$).

For practical problems it may be  useful to maximize the minimum power among all alternatives of interest, so  the minimax objectives for a range of possible parameter settings  $\vartheta \in \Theta_{\Pi}\subset \Theta$ corresponding to $\Pi_{any},\; \Pi_{total}$ are, respectively \citep{Romano11, Rosenblum14, Rosset18}: 
$$
\min_{\vartheta \in \Theta_{\Pi}} \mP_{\vartheta} \left(\sum_{k=1}^Kh_kD_k > 0\right); \quad \min_{\vartheta \in \Theta_{\Pi}} \mE_{\vartheta} \left(\sum_{k=1}^Kh_kD_k \right).
$$
For example, if $h_k=0$ for $\theta_k\geq 0$ and $h_k=1$ for $\theta_k<0$, $k=1,\ldots, K$, then it may be desirable to obtain the maximin solution if the alternative parameter value is at most $-2$ for every coordinate, so
$$
\min_{\vartheta \in \Theta_A} \mP_{\vartheta} \left(\sum_{k=1}^KD_k > 0\right); \quad \min_{\vartheta \in \Theta_A} \mE_{\vartheta} \left(\sum_{k=1}^KD_k \right)
$$
for $\Theta_A = \{\vartheta: \theta_k(\vartheta)\leq -2, k=1,\ldots,K \}$.

Another important model is the two group model, pioneered by Efron  \citep{Efron01}.
In our framework this model can be formulated as an implicit prior distribution over $\vartheta,$ by assuming that each $h_k\sim Ber(\pi)$ i.i.d, and $\theta_k | h_k=0$ is fixed at a certain value (say $ \theta_{k0}$), and similarly   $\theta_k | h_k=1$ is fixed at another value (say  $\theta_{k1})$
. 
In this setting the common power function considered is: 
 $$\Pi_{total}(\bD) = \pi\sum_{k=1}^K\mP(D_k=1\mid h_k=1).$$ 
An extension in the spirit of maximin may be to assume that the parameters $\theta_{k0},\theta_{k1}$ are unknown, but our interest is in a predefined range of parameter values, and then the maximin objective is to maximize the following: 

$$
\min_{\{(\theta_{k1}, \theta_{k0}) \in  (-\infty, -2]\times [0,\infty), \ k=1,\ldots,K\}}  \sum_{k=1}^K\mP(D_k=1\mid h_k=1).
$$

\subsection{The Constraints}\label{subsec-constraints}

The most common error measure for control over false discoveries in clinical trials is the FWER, which is the probability of falsely rejecting at least one true null hypothesis: 
$
\textrm{FWER}_{\vartheta} = \mP_{\vartheta}\left(\sum_{k=1}^K(1-h_k)D_k>0 \right).
$
This is the measure of error recommended  by the FDA \citep{FDA17} in their  comprehensive guidance on handling multiple endpoints in clinical trials. The recommendation is to control the FWER in the strong sense, i.e., 
for every possible parameter vector (including the elements that are not null), at a pre-specified level $\alpha$: $\textrm{FWER}_{\vartheta}\leq \alpha \ \forall \ \vartheta \in \Theta. $

%

Another  error measure is the false discovery rate (FDR), which is very popular when many hypotheses are simultaneously examined. The  FDR control in the strong sense is satisfied at level $\alpha$ if: 
$
\textrm{FDR}_{\vartheta} = \mE_{\vartheta}\left(\frac{\sum_{k=1}^K(1-h_k)D_k}{\max(\sum_{k=1}^K D_k,1)} \right)\leq \alpha \ \forall \  \vartheta \in \Theta.
$
A weighted version of this error rate was recommended in \cite{Benjamini17} for studies with multiple primary and secondary endpoints in clinical trials. 
The FDR can be controlled in the strong sense \citep{Benjamini95} or under a pre-specified data generation process such as Efron's ``two-group model" \citep{Efron01}. 

We focus on  strong FWER control in this work, since our interest is in moving from $K=1$ to $K>1$ in clinical trials, where $K$ is usually still very small.

\subsection{Restrictions on the decision rule}\label{subsec-restrictions}

In this section we discuss ``common sense'' properties that we may wish our  procedures to satisfy. 
It will be convenient to describe these properties (or restrictions) in terms of $p$-values rather than the data. 
The problem at hand is that we observe $K$ $p$-values for the $K$ tests, denoted by $\bp=(p_1,\ldots,p_K)$, and based on these values the decision rule indicates which null hypotheses are rejected and which are not. 

One property may be to reject only hypotheses with $p$-values at most $\alpha$, since intuitively when facing multiplicity the rejection threshold should be adjusted to be more severe than when only a single hypothesis is tested. This restriction appears when  enforcing consonance  in closed testing procedures \citep{Romano11}.
We say a procedure is 
 {\em marginally nominal} $\alpha$  if it satisfied $D_i(\bp) = 0$ for $p_i>\alpha$, $i=1,\ldots,K$. 


Another property is a restriction that is  logical to impose when considering pairs of vectors of $p$-values, $\bp$ and $\bq$. This restriction appears for the optimization of exact tests in \cite{Ristl18}. Let $\succeq, \preceq$ symbolize that  the partial order relations are satisfied if the inequality holds for every coordinate. We call a procedure 
{\em weakly monotone} if $\bD(\bp) \succeq \bD(\bq)$ whenever  $\bp\preceq \bq$. 
The procedures depicted in Figure \ref{fig-rejectionregion} and \ref{fig-rejectionregion2} are all weakly monotone, since the slopes are negative and there are no gaps, so that if $p_1<q_1$, then $\{p: D_2(q_1,p ) =1\} \subseteq \{p: D_2(p_1,p ) =1\}$. 


\subsection{Formulation as an optimization problem}
To formulate the  problem as an optimization problem, we need to select the objective, and define the constraints we wish to impose. 
 We denote any of the power functions discussed in \S~\ref{subsec-objective}  generically by  $\Pi(\bD)$. We write the optimization  problem  of finding the test with optimal power, subject to strong FWER  control, as an infinite dimensional binary program (i.e., decide which hypotheses to reject for every realized vector of $p$-values): 
 
 \begin{eqnarray}\label{opt}
\max_{\vec D:[0,1]^K\rightarrow \{0,1\}^K} && \Pi(\bD)\\
\mbox{s.t. } &&  \textrm{FWER}_{\vartheta}(\bD) \leq \alpha\;,\; \forall  \vartheta\in \Theta.
 \nonumber
\end{eqnarray}
 
 We can also include the problem of optimal sample size allocation for a given total sample size $N = \sum_{k=1}^K n_k$, where  $n_k$ observations are used to compute the $p$-value for the $k$th null hypothesis. This problem arises when the $K$ hypotheses refer to $K$ treatment groups, or to $K$ different populations (in other settings finding optimal sample sizes is more complicated).  The solution is found by solving \eqref{opt} for different values of  $\bn = (n_1,\ldots,n_k)\in \{ (n_1,\ldots,n_K): \sum_{k=1}^K n_k=N, n_k\geq 0, n_k\in \mathbb N,  k=1,\ldots,K \}$, and searching for  the value of $\bn$ that maximizes the objective.

We are interested in enforcing the {\it marginally nominal $\alpha$} property: 
 \begin{eqnarray}\label{opt4-nominal}
\max_{\bD:[0,1]^K  \rightarrow \{0,1\}^K} && \Pi
(\bD)\\
\mbox{s.t. } &&   \textrm{FWER}_{\vartheta}(\bD) \leq \alpha\;,\; \forall  \vartheta\in \Theta; \quad D_i(\vec p) = 0 \ \textrm{if } \ p_i>\alpha, \ i=1,\ldots,K. \nonumber 
\end{eqnarray}

Or in enforcing  {\it weak monotonicity}: 
 \begin{eqnarray}\label{opt4}
\max_{\bD:[0,1]^K  \rightarrow \{0,1\}^K} && \Pi
(\bD)\\
\mbox{s.t. } &&  \textrm{FWER}_{\vartheta}(\bD) \leq \alpha\;,\; \forall  \vartheta\in \Theta; \quad \bD(\bp)\succeq \bD(\bq), \ \forall \bp\preceq \bq. 
 \nonumber
\end{eqnarray}

\section{The case of $K=2$ hypotheses}\label{sec-K2}

In \S~\ref{subsec-K2-nominal} we solve the OMT problem 
under the marginally nominal $\alpha$ requirement, problem \eqref{opt4-nominal}. 
In \S~\ref{subsec-K2-weaklymonotone} we show that the same solution applies under the weakly monotone requirement, problem \eqref{opt4}, when some assumptions are added. 
In \S~\ref{subsec-connections} we discuss power objectives that result in Hommel's procedure and 
in the procedure suggested by \cite{Bittman09}.

For $K=2$ hypotheses, the power objectives in \S~\ref{subsec-objective} can all be expressed in the following simple form: 
\begin{equation}
\Pi(\bD) = \int (D_1(\bp)  a_1(\bp)+D_2(\bp)  a_2(\bp)+D_3(\bp)  a_3(\bp) )d \bp,\label{eq-unified-expression}
\end{equation}
 where $D_3(\bp) = \max(D_1(\bp), D_2(\bp ))$ and the functions $a_i(\bp), i=1,2,3$ are objective-specific functions of the $p$-values.

Next, we provide the $a_i(\bp)$'s for three specific objectives, that we shall denote by specific names, as they will be referred to henceforth in our numerical investigations. 
For simplicity, we assume that the bivariate density of $\bp$ is continuous, and  that  the marginal distribution of the $p$-values under the null hypothesis (i.e., when $h_k=0$) is uniform over the unit interval (to be a valid $p$-value, its null distribution has to be stochastically at least as large as the uniform, and for continuous test statistics it is often uniform at the parameter value that separates the null from the alternative, \citealt{Lehmann05b}).

We need the following additional notation: 
for $\vartheta\in \Theta$, let $g_{\vartheta}(\bp)$ denote the joint density of the $p$-values, and let $g_{1}(p_1) = \int_0^1g_{\vartheta}(\bp)dp_2$ and $g_{2}(p_2) = \int_0^1g_{\vartheta}(\bp)dp_1$ be the two marginal densities. 

Considering a  fixed $\vartheta\in \Theta$ such that $\bh(\vartheta) = \bone$, $\Pi_{any}(\bD)$ is expressed by plugging into \eqref{eq-unified-expression} $a_1(\bp)=a_2(\bp)=0$ and $a_3(\bp) = g_{\vartheta}(\bp)$. We denote  this objective henceforth as $\Pi_{any}$, in agreement with previous literature \citep{Bittman09, Romano11}. $\Pi_{total}(\bD)$ is expressed by plugging into \eqref{eq-unified-expression} $a_1(\bp)=a_2(\bp)=g_{\vartheta}(\bp)$ and $a_3(\bp) = 0$.  We denote the objective $\Pi_{total}(\bD)/2$, i.e., the expected average number of discoveries (which are necessarily true discoveries), as $\Pi_{avg}$.

If with probability 1/2, $h_i=1$ with $p$-value density $g_{\theta_i}(p_i)$, and otherwise $h_i=0$ with a uniform $p$-value density, then: $\Pi_{any}(\bD)$ is expressed by plugging into \eqref{eq-unified-expression} $a_i(\bp)=g_{i}(\bp)/4$ for $i=1,2$  and $a_3(\bp) = g_{\vartheta}(\bp)/4$;  $\Pi_{total}(\bD)$ is expressed by plugging into \eqref{eq-unified-expression} $a_i(\bp)=g_{i}(\bp)/2$ for $i=1,2$  and $a_3(\bp) = 0$. We denote by $\Pi_1$ the objective $\Pi_{total}(\bD)$ for this choice of prior probability on the hypotheses states (the subscript one in $\Pi_1$ is used since this objective also corresponds  to the average number of rejections if exactly one null hypothesis is false, and a-priori each has probability $1/2$ of being the false null hypothesis).

\subsection{The marginally nominal $\alpha$ OMT procedure for strong FWER control}\label{subsec-K2-nominal}

For maximizing power,  we need to consider the power objective's integrand in \eqref{eq-unified-expression}, $D_1(\bp)  a_1(\bp)+D_2(\bp)  a_2(\bp)+D_3(\bp)  a_3(\bp)$. Since $D_i(\bp)\leq \mI(p_i\leq \alpha)$, the integrand is considered with $\bD$ replaced with  $(\mI(p_1\leq \alpha), \mI(p_2 \leq \alpha))$, call this the {\em score} function
  $$s(\bp) = \mI(p_1\leq \alpha) \times a_1(\bp)+\mI(p_2\leq \alpha)\times  a_2(\bp)+\mI(\max(p_1,p_2)\leq \alpha)  a_3(\bp).$$

If  one of the null hypotheses is true and the other is false then the FWER constraint is satisfied, since 
$\{\bp: D_i(\bp)=1\}\subseteq \{\bp: p_i\leq \alpha \}$. Thus the only binding integral constraint is the global null constraint $\textrm{FWER}_{\vartheta_0}\leq \alpha$, where ${\vartheta_0} $ is the parameter vector for which both $p$-values have a marginal uniform distribution.  
 The area for rejection in the unit square has to be $\alpha$, but the decision which realizations $\bp$ to include in the rejection area are solely driven by $s(\bp)$. Hence the algorithm: 
 
\begin{enumerate}
\item Consider for rejection only vectors $\bp$ for which at least one $p$-value  does not exceed $\alpha$.
\item Further retain only the vectors $\bp$  which give the highest benefit in power. This is done by finding the threshold $t$
for which $\mP_{\vartheta_0}(s(\bp)>t) = \alpha$, that is $t$ is such that rejecting all points with scores above it, the FWER constraint (at the global null $\vartheta_0$) is exactly $\alpha$.  For each retained vector, the decision will be to reject at least one hypothesis. 
    \item Reject all hypotheses with a $p$-value that does not exceed $\alpha$ among the retained vectors $\bp$ in Step (2), i.e.,  if both $p$-values are at most $\alpha$ reject both hypotheses, otherwise reject only the hypothesis with the smaller $p$-value. 
\end{enumerate}

More concisely, the OMT procedure for strong FWER control at level $\alpha$ is: 
\begin{equation}\label{eq-novelpolicy-1}
    D^*_i(\vec p) = \mI(p_i\leq \alpha)\times \mI( s(\bp)>t(\alpha)), i=1,2,
\end{equation}
where $t(\alpha)$ is the threshold that satisfies  
\begin{eqnarray}\label{eq-novelpolicy-2}
\int_{0}^1\int_0^1\max(D^*_1(u,v), D^*_2(u,v))g_{{\vartheta_0}}(u,v)du dv= \alpha.
\end{eqnarray}  
For independent $p$-values $g_{\vartheta_0}(u,v)=1$.

The marginally nominal $\alpha$ OMT procedure is formalized in the following theorem. 

 \begin{theorem}\label{prop3-K2}
 Assuming the joint density of the $p$-values is continuous, then among marginally nominal $\alpha$ procedures, 
the objective is maximized with strong FWER control for  procedure
\eqref{eq-novelpolicy-1}-\eqref{eq-novelpolicy-2},
where $s(\bp)$ is the objective's integrand with $\bD$ replaced by $(\mI(p_1\leq \alpha), \mI(p_2 \leq \alpha))$.
\end{theorem} 
\begin{proof}
Since only a hypothesis with $p$-value at most $\alpha$ can be rejected,  it follows that  the optimal solution is in the restricted domain $\{(p_1,p_2): p_1\leq \alpha \ \textrm{or} \ p_2\leq \alpha \}$. This restriction guarantees $\textrm{FWER}_{\vartheta}\leq \alpha$ for all $\vartheta\in \Theta$  with $h_1+h_2=1$.  
The global null constraint is also satisfied since $t$ is set so that \eqref{eq-novelpolicy-2} is satisfied. The threshold $t(\alpha)$ necessarily exists and is unique since the left hand side of \eqref{eq-novelpolicy-2} defined for  a general threshold $t\geq 0$ in \eqref{eq-novelpolicy-1} is a function of $t$ that satisfies the following properties: (1) it is continuous and is decreasing in $t$; (2) it is $\geq \alpha$ for $t=0$;  and (3) it is 0 for $t\rightarrow \infty$.   
 Thus, strong FWER control is satisfied with this solution, and all that remains is to show that there does not exist another procedure in the restricted domain that is more powerful. 
 The rejection area with $\bD^*$ includes the highest values of $s(\bp)$ within the restricted domain $\{(p_1,p_2): p_1\leq \alpha \ \textrm{or} \ p_2\leq \alpha \}$. 
 If the ordering is not by $s(\bp)$, the procedure is necessarily sub-optimal by an argument similar to that used to prove the Neyman-Pearson lemma \citep{Lehmann05b}.
\end{proof}

\subsection{The weakly monotone OMT procedure for strong FWER control}\label{subsec-K2-weaklymonotone}

We aim to solve problem \eqref{opt4} for $K=2$. 
We start by adding the following assumption on the parameter domain: if  null hypothesis $i$ is false,  
for any $u>0$ and $\epsilon>0$ there exists a $\vartheta_i(u,\epsilon)$ such that 
for all $\vartheta\in \Theta$ with $\theta_i(\vartheta)<\theta_{i}(\vartheta_i(u,\epsilon))$,
$P_\vartheta(p_i \le u) \ge 1-\epsilon$. 

This restriction can be considered as a weaker version of the standard 1-sided alternative setting, where the set of considered alternatives includes extreme values that yield $p$-values which are arbitrarily close to zero. To demonstrate a specific setting, consider the case where the two test statistics are bivariate normal with correlation $\rho$. Formally, this implies that the two p-values are random variables according to
$$    p_1=\Phi(\theta_1+Z_1) \ \textrm{and} \   p_2=\Phi(\theta_2+\rho Z_1+\sqrt{1-\rho^2}Z_2)$$
where $Z_i$ are iid standard normals, $\Phi$ their cumulative distribution function. Under the null the values of $\theta_i$ are zero and under the alternative they are negative.  If  $\theta_i=0$, the marginal distribution for $p_i$ is uniform, and as $\theta_i$ becomes more negative the distribution of $p_i$ becomes concentrated near 0 as required by the assumption, provided that the range of $\theta_i$ stretches to $-\infty$. The general formulation in \eqref{opt4} allows for cases where the plausible values are restricted to say $\theta_i>-2$ and then the above assumption and ensuing results do not follow, but typically the problem formulation does not include an upper bound on the power (or lower bound on the parameter space) for the normal means problem. 

 \begin{theorem}
Under the distributional assumptions above, if  $a_i(\bp)$ in the score function is non-increasing in each coordinate for $i=1,2,3$, then
the solution to problem \eqref{opt4}  coincides with that of problem \eqref{opt4-nominal}, i.e., the optimal procedure is the one stated in Theorem \ref{prop3-K2}.
\end{theorem}
\begin{proof}
 We shall make use of the following  Lemma, which is proved in  Appendix B. 
 \begin{lemma}\label{prop1-K2}
Under the distributional assumptions above, 
the solution to problem \eqref{opt4} satisfies $ D_i^*(\bp) =0$ if $ p_i> \alpha$, for $i=1,2$.
\end{lemma}

The lemma implies that the solution to problem \eqref{opt4} will coincide with that of problem \eqref{opt4-nominal}, i.e., the optimal procedure is the one stated in Theorem \ref{prop3-K2}, if the resulting procedure is weakly monotone. 
 The  procedure stated in Theorem \ref{prop3-K2} is indeed weakly monotone  since   $s(\bp)$ is non-increasing in each coordinate if $a_i(\bp)$ is non-increasing in each coordinate for $i=1,2,3$. Thus, if $ \bq\succeq \bp$ then $s(\bp)\geq s(\bq)$, and $\bp$ will enter the rejection set before $\bq.$
\end{proof}

\begin{remark}
For the objectives considered in \S~\ref{subsec-objective}, $a_i(\bp)$ is non-increasing in each coordinate for $i=1,2,3$ for independent $p$-values, if their marginal densities are non-increasing. This is the case when the $p$-values  come from test statistics that satisfy the monotone likelihood ratio (MLR). For example,  any one-parametric exponential family fulfils this MLR property with respect to its sufficient statistic \citep{Lehmann05}. 
\end{remark}

\subsection{Power objectives that result in existing procedures }\label{subsec-connections}
We consider the one-sided normal means problem ($\theta_i\geq 0$ if $h_i=0$ and $\theta_i< 0$ if $h_i=1$), when the sample sizes are the same for each hypothesis and the 
non-null distribution of $p_i$ is the same for $i=1,2$ for the data generation implied by the objective. For simplicity, we assume the $p$-values are independent, so the score $s(\bp)$ simplifies to the following
: 
\begin{eqnarray}
&& \mI(\max(p_1,p_2)\leq \alpha)
 \times e^{\Phi^{-1}(p_1) \theta-\frac{\theta^2}{2}}\times  e^{\Phi^{-1}(p_2) \theta-\frac{\theta^2}{2}}
 \  \textrm{ for }  \ \Pi_{any} \label{eq-piany}\\ && \frac12(\mI(p_1\leq \alpha)+\mI(p_2\leq \alpha))\times e^{\Phi^{-1}(p_1) \theta-\frac{\theta^2}{2}}\times  e^{\Phi^{-1}(p_2) \theta-\frac{\theta^2}{2}} \  \textrm{ for }  \ \Pi_{avg}  \label{eq-pitotal-fixed} \\ &&  \frac12(\mI(p_1\leq \alpha)\times e^{\Phi^{-1}(p_1) \theta-\frac{\theta^2}{2}}+\mI(p_2\leq \alpha)\times  e^{\Phi^{-1}(p_2) \theta-\frac{\theta^2}{2}})\  \textrm{ for }  \ \Pi_{1},  \label{eq-pi1}
\end{eqnarray}

 were $\Phi^{-1}$ is the quantile function of the standard normal distribution. 

We point out three interesting connections of  the weakly monotone OMT procedure to existing procedures: 1)  Hommel's procedure is the OMT procedure for objective $\Pi_1$,  when the difference between the alternative and the (boundary) null parameter is not too large; 2)  Bittman's consonant improvement over closed-Stouffer \citep{Bittman09} is the weakly monotone OMT procedure for objective $\Pi_{any}$ ; 3) Without the weak monotonicity constraint, the general
OMT procedure for $\Pi_{any}$ is to reject the minimal $p$-value if the intersection hypothesis is rejected at level $\alpha$ using Stouffer's test \citep{Rosenblum14b, Rosset18}, so the weakly monotone OMT procedure 
is necessarily less powerful than this OMT procedure. Specifics follow. 

Hommel's procedure is $D_i = \mI(p_i\leq \alpha/2 \cup \max(p_1,p_2)\leq \alpha), i=1,2.$  If the score $s(\bp)$ is smaller outside Hommel's rejection region  than the score inside it, i.e., 
\begin{equation}\label{eq-hommelequiv}
s(\bp)<s(\alpha,\alpha) \ \forall \ \bp \in \left((\frac \alpha 2, \alpha]\times [\alpha,1]\right) \cup \left([\alpha,1]\times (\frac \alpha 2, \alpha]\right),
\end{equation}
then since Hommel's rejection region is exactly $\alpha$ when both null hypotheses are true, no additional points can be added to the rejection region of the weakly monotone OMT procedure while still maintaining strong FWER control at level $\alpha$. Hence the weakly monotone OMT procedure coincides with Hommel's procedure. For normal means it is easy to see that for a fairly wide range of $\theta$'s  the inequality in \eqref{eq-hommelequiv} is satisfied   for $\alpha/2<p_1\leq \alpha,\; p_2>\alpha$: since $s(\alpha,\alpha) =  e^{\Phi^{-1}(\alpha) \theta-\frac{\theta^2}{2}}$, then $s(p_1,p_2) =\frac 12 e^{\Phi^{-1}(p_1) \theta}\leq \frac 12 e^{\Phi^{-1}(\alpha/2) \theta-\frac{\theta^2}{2}}$, and 
$$ e^{\Phi^{-1}(\alpha) \theta-\frac{\theta^2}{2}}>\frac 12 e^{\Phi^{-1}(\alpha/2) \theta-\frac{\theta^2}{2}} \ \iff \ \theta> \frac{-\log 2}{\Phi^{-1}(\alpha)-\Phi^{-1}(\alpha/2)}  .$$ For example, if $\alpha = 0.025$ then the weakly monotone OMT procedure coincides with Hommel for $\theta>-2.46$. Note that for stronger alternatives (i.e., $\theta<-2.46$)  the rejection region varies with $\theta$.  

The closed-Stouffer procedure is $D_i = \mI(p_i\leq \alpha)\times \mI( \Phi^{-1}(p_1)+\Phi^{-1}(p_2)\leq \sqrt{2} \Phi^{-1}(\alpha)), i=1,2$. This procedure is sub-optimal for $\Pi_{any}$ since the FWER is controlled at a level  smaller than $\alpha$. \cite{Bittman09} suggested instead the procedure $D_i = \mI(p_i\leq \alpha)\times \mI( \Phi^{-1}(p_1)+\Phi^{-1}(p_2)\leq t(\alpha)), i=1,2,$ where $t(\alpha)$ is such that the rejection probability at the (boundary) null parameter value (i.e., when $\theta=0$) is exactly $\alpha$. This procedure dominates closed-Stouffer since $t(\alpha)> \sqrt{2} \Phi^{-1}(\alpha), $ and it is identical to the weakly monotone OMT procedure for objective $\Pi_{any}$. Thus,  the weakly monotone OMT procedure does not vary with $\theta$, so for two one-sided normal alternatives it is uniformly most powerful among all marginally nominal $\alpha$ or weakly monotone procedures. 

The OMT procedure for $\Pi_{any}$  is to reject the smallest $p$-value if $\Phi^{-1}(p_1)+\Phi^{-1}(p_2)\leq \sqrt{2} \Phi^{-1}(\alpha)$ \citep{Rosenblum14b, Rosset18}. This procedure is not weakly monotone, since, for example, with realization $\bp = (\alpha/2+\alpha^2/4, \alpha/2-\alpha^2/4)$ only the second hypothesis is rejected, and with realization   $\bq = (\alpha/3-\alpha^2/4, \alpha/3+\alpha^2/4)$ only the first hypothesis is rejected. Thus, even though $\bq\preceq \bp$ for $\alpha<1/3$, the weakly monotone requirement that $\bD(\bq)\succeq \bD(\bp)$ is violated.  Moreover, the OMT procedure rejects the hypothesis with minimal $p$-value even if its value is greater than $\alpha$, as long as $\Phi^{-1}(p_1)+\Phi^{-1}(p_2)\leq \sqrt{2} \Phi^{-1}(\alpha)$, so it is not a marginally nominal $\alpha$ procedure. So for $\Pi_{any}$, the uniformly most powerful procedure 
is more powerful than the solution to problem
\eqref{opt4-nominal} or \eqref{opt4}.

\section{Numerical examples}\label{sec-numerical}
We consider the following objectives: when both null hypotheses are false, the probability of at least one true discovery ($\Pi_{any}$), and the average expected number of true discoveries ($\Pi_{avg}$); the expected total number of true discoveries when the two group model prior on $h_i$ is $\sim Ber(1/2)$, $i=1,2$  ($\Pi_1$); finally, a combination of objectives, $1/3\times \Pi_{any} + 2/3 \times \Pi_1$, which coincides with $\mP(h_1D_1+h_2D_2>0\mid h_1+h_2>0)$ when the two group model prior is $Ber(1/2)$.

We base the examples on the context provided in \S~\ref{subsec-APEX}, by assuming a baseline  event rate of 7.5\% and an expected relative reduction of 35\% for each group when the drug that defines the group is effective. We assume first that the  sample sizes are the same in each group. In this  exchangeable setting  the two hypothesis testing problems are identical. We consider next the optimal sample allocation for a total sample size $N$, divided into $r\times N$ patients in the first group and $(1-r)\times  N$ patients in the second group. An intuitive guess is that since the expected relative reduction is the same across groups,  then the optimal allocation is an equal split, i.e., $r=1/2$. However, when $N$ is small, so that the probability of rejecting a hypothesis is small even when all $N$ samples are allocated to a single hypothesis, we find that $r=1/2$ is sub-optimal.

Table \ref{tab-powerexchangeable} shows the power comparison for $r=1/2$ and $N=4800$. In this ``strong signal'' exchangeable setting, Hommel's procedure is not the weakly monotone OMT procedure for objective $\Pi_1$, but it is still almost as powerful as the optimal procedure.  
However, with objective $\Pi_{any}$ the power advantage over Hommel's procedure is more than $4\%$. Arguably, the procedure with the most satisfactory power properties is with objective $1/3\times \Pi_{any} + 2/3 \times \Pi_1$, since it dominates Hommel by more than $3\%$ for $\Pi_{any}$, and its power is only  $0.4\%$ lower than  Hommel's for $\Pi_1$. This procedure is fairly close to optimal for power measures $\Pi_{avg}, \Pi_1,$ and $\Pi_{any}$, in addition to being optimal for its own objective. 

 \begin{table}[ht]
 \centering
 \caption{\label{tab-powerexchangeable} For exchangeable hypotheses, each with power 85\% of being detected on its own (corresponding to $n=1200$ subjects in each of the four arms in the APEX study described in \S~\ref{subsec-APEX}),  for each procedure (column) with strong control FWER guarantee at the $\alpha=0.025$ level, four measures of power (rows).   The OMT procedure for $\Pi_{avg}$ coincides with that of $\Pi_{any}$, so they are presented in the same column.  In bold largest power in the row. 
 }
 \begin{tabular}{lrrrrr}\hline
  & \multicolumn{3}{c}{OMT for Objective} & \multicolumn{2}{c}{Off-the-shelf competitor} \\
 Power measure & $\Pi_{avg}/\Pi_{any}$ &   $\Pi_{1}$ & $1/3\times \Pi_{any}  $ & closed Stouffer & Hommel \\ 
  &&&$+2/3 \times \Pi_1$&&\\
   \hline
 $\Pi_{avg}$  & {\bf 0.747} & 0.725  & 0.741 & 0.744 & 0.725\\ 
  $\Pi_{any}$  & {\bf 0.928} & 0.885  & 0.916 & 0.921 & 0.885\\ 
    $\Pi_{1}$ & 0.557 & {\bf 0.670}  & 0.665 & 0.448 & 0.670 \\ 
    $1/3\times \Pi_{any} + 2/3 \times \Pi_1$ & 0.681 & 0.741  & {\bf 0.749} & 0.606 & 0.741\\ 
   \hline    
\end{tabular}
\end{table}

Table \ref{tab-splitsample} compares the power with $r=1/2$ and $r=1/4$ for a large total sample size of $N=4800$ and for a small total sample size of $N=600$. For the larger sample size, the power is greatest for $r=1/2$ for all objectives but $\Pi_{any}$. However, when the sample size is small, $r=1/4$ has the largest power for all objectives.  Thus it seems that for all power measures except $\Pi_{any}$, for the level of power typically desired in clinical trials, an equal split is preferred for $K=2$, but this may not be the case for other applications that may have low  power.

\begin{table}[ht]
 \centering
 \caption{\label{tab-splitsample} For an equal and unequal split of sample size, the power of the weakly monotone OMT procedure as well as closed-Stouffer and Hommel. Each procedure guarantees strong FWER control at the $\alpha=0.025$ level. In bold largest power in the row. The OMT procedure for $r=1/2$ and $N=4800$ for objective $\Pi_1$ has a slightly higher power than Hommel, but they have the same value up to the third decimal point.  
 }
 \begin{tabular}{llcccccc}
 \hline
&   & \multicolumn{3}{c}{$r=1/2$} & \multicolumn{3}{c}{$r=1/4$} \\
Total  &   Power  & OMT   &   closed  & Hommel &  OMT &   closed  & Hommel \\ 
sample size &    measure &    &   Stouffer &  &  &    Stouffer &  \\ 
 
   \hline
 $N=4800$ & $\Pi_{avg}$   &{\bf 0.747}& 0.744 & 0.725 &0.684& 0.675 & 0.665\\ 
  & $\Pi_{any}$   &0.928& 0.921 & 0.885& {\bf 0.943} & 0.925 & 0.905\\ 
    &$\Pi_{1}$  &{\bf 0.670} & 0.448 & 0.670&0.611& 0.426 & 0.608\\ 
    &$1/3\times \Pi_{any} + 2/3 \times \Pi_1$  & {\bf 0.748} & 0.605 & 0.740&0.713 & 0.592 & 0.707\\ 
   \hline    
   $N=600$ & $\Pi_{avg}$   & 0.132 & 0.112 & 0.106 &{\bf 0.133}& 0.106 & 0.106 \\ 
  & $\Pi_{any}$   &0.240 & 0.200 & 0.189 &{\bf 0.246}  & 0.191 & 0.192 \\ 
    &$\Pi_{1}$  &0.099 & 0.060 & 0.099 &{\bf 0.109} & 0.061 & 0.100\\ 
    &$1/3\times \Pi_{any} + 2/3 \times \Pi_1$  &  0.141 & 0.107 & 0.129&{\bf 0.150} & 0.105 & 0.131 \\  
  \hline 
\end{tabular}
\end{table}

For $\Pi_{any}$ the optimal split is always $r\in \{0,1\}$. This follows since the most powerful test of the global null against the alternative that both null hypotheses are false is the same regardless of $r\in[0,1]$, and $\Pi_{any}$ is upper bounded by this power. This bound is only  achieved when $r=0$ or $r=1$ exactly. This result is formalized in a proposition in  Appendix C. It suggests that  $\Pi_{any}$ is not a suitable objective when considering optimizing $r$, since, unlike all other measures, it is unaffected by which or how many null hypotheses are rejected provided that  at least one is rejected.

Finally,  we computed the relative saving in comparison to using Hommel's procedure to achieve  that same power, i.e., $\frac{N_{Hommel}-4800}{N_{Hommel}}\times 100$ for the OMT power achieved for each of the four power measures  with $N=4800$ (Table \ref{tab-splitsample}, column 3, rows 1--4).  The greatest  saving is with regard to power measure $\Pi_{any}$, 9.91\% (since 5328 subjects are needed to achieve at least as much power as with the OMT procedure). With power measures  $\Pi_{avg}$ and $1/3\times \Pi_{any} + 2/3 \times \Pi_1$ the relative saving is only    1.56\% and 0.74\%, respectively (since 4876 and 4836   subjects are needed to achieve the OMT power for $\Pi_{avg}$ and $1/3\times \Pi_{any} + 2/3 \times \Pi_1$, respectively).

\section{Optimal design for the APEX trial}\label{sec-APEX}

We return to the APEX trial introduced in \S~\ref{subsec-APEX}. 
In addition to strong FWER control, we required that the procedure be weakly monotone, since it seems reasonable to require that a smaller pair of $p$-values should lead to at least as many rejections as a larger pair of $p$-values.
So the resulting OMT procedure is the  one  developed  in \S~\ref{sec-K2}. Figure \ref{fig-rejectionregion2} shows the OMT procedures for each objective. The sample size is bigger for group 1 than for group 2, so the rejection procedure is not symmetric in the $p$-values (unlike the symmetric procedures of the competitors Hommel and closed-Stouffer). 

Table \ref{tab-powerexample-full} shows a power comparison between these procedures, where we clearly see that which is best depends on the   power objective (which implies a data generation mechanism). A comparison of the two off-the-shelf procedures shows that if both null hypotheses are false, the expected number of true discoveries, as well as the probability of rejecting at least one false null, is slightly larger for closed-Stouffer, but  if only one null is false, the probability of discovering it is far greater using Hommel's procedure. 
The OMT procedure for $1/3\times \Pi_{any} + 2/3 \times \Pi_1$ dominates Hommel. This procedure is also a close second to any of the other three procedures, i.e., it is better than the OMT procedure for $\Pi_1$ when optimizing $\Pi_{any}$ or $\Pi_{avg}$, and it is much better than the OMT procedure for $\Pi_{any}$ or $\Pi_{avg}$ when optimizing $\Pi_1$. Arguably, if  it is unclear which objective is desired the OMT procedure for $1/3\times \Pi_{any} + 2/3 \times \Pi_1$ seems flexible to the underlying state of the problem and hence it appears to be  the preferred choice. Hommel seems to be a better off-the-shelf procedure than closed Stouffer. Moreover, the OMT procedure also provides a benchmark as to how well a procedure can possibly perform; for example, it informs the researcher that using a closed Stouffer procedure if the objective is $\Pi_1$ is not desirable as there is a loss of power of $0.20$. 

 \begin{table}[ht]
 \centering
 \caption{\label{tab-powerexample-full} For non-exchangeable hypotheses, each with the realized sample size in the APEX study and the  effect used for the sample size calculation (described in \S~\ref{subsec-APEX}), for each procedure (column) with strong control FWER guarantee at the $\alpha=0.025$ level, four measures of power (rows). The OMT procedure for $\Pi_{avg}$ coincides with that of $\Pi_{any}$, so they are presented in the same column. In bold largest power in the row.
 }
 \begin{tabular}{lrrrrr}
 \hline
  & \multicolumn{3}{c}{OMT for Objective} & \multicolumn{2}{c}{Off-the-shelf competitor} \\
 Power measure & $\Pi_{avg}/\Pi_{any}$ & $\Pi_{1}$ & $1/3\times \Pi_{any}  $ & closed Stouffer & Hommel \\ 
 & &&$+2/3 \times \Pi_1$&&\\
   \hline
 $\Pi_{avg}$  & {\bf 0.833}  &  0.822 & 0.829 & 0.832 & 0.774 \\ 
  $\Pi_{any}$  & {\bf 0.970}  & 0.949 & 0.963 & 0.968 & 0.953 \\ 
    $\Pi_{1}$ & 0.649 &{\bf 0.777}   & 0.775 & 0.555 & 0.774 \\ 
    $1/3\times \Pi_{any} + 2/3 \times \Pi_1$ & 0.756  & 0.835 & {\bf 0.838}  & 0.693 & 0.833 \\ 
   \hline    
\end{tabular}
\end{table}

\begin{figure}[htbp]
 \begin{tabular}{ccc}
\hspace{-1cm}     \includegraphics[width=4.5cm,height=4.5cm]{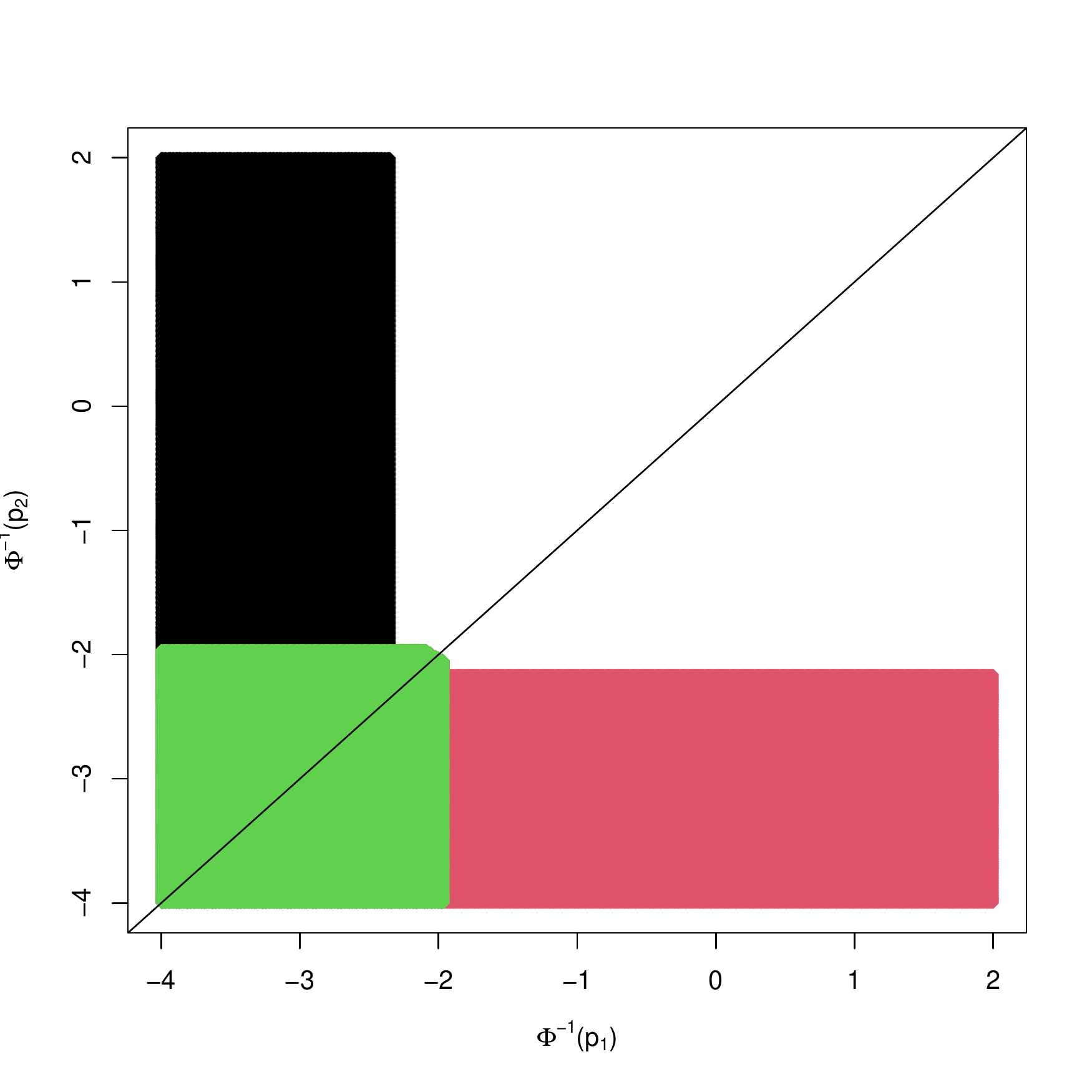} &   \includegraphics[width=4.5cm,height=4.5cm]{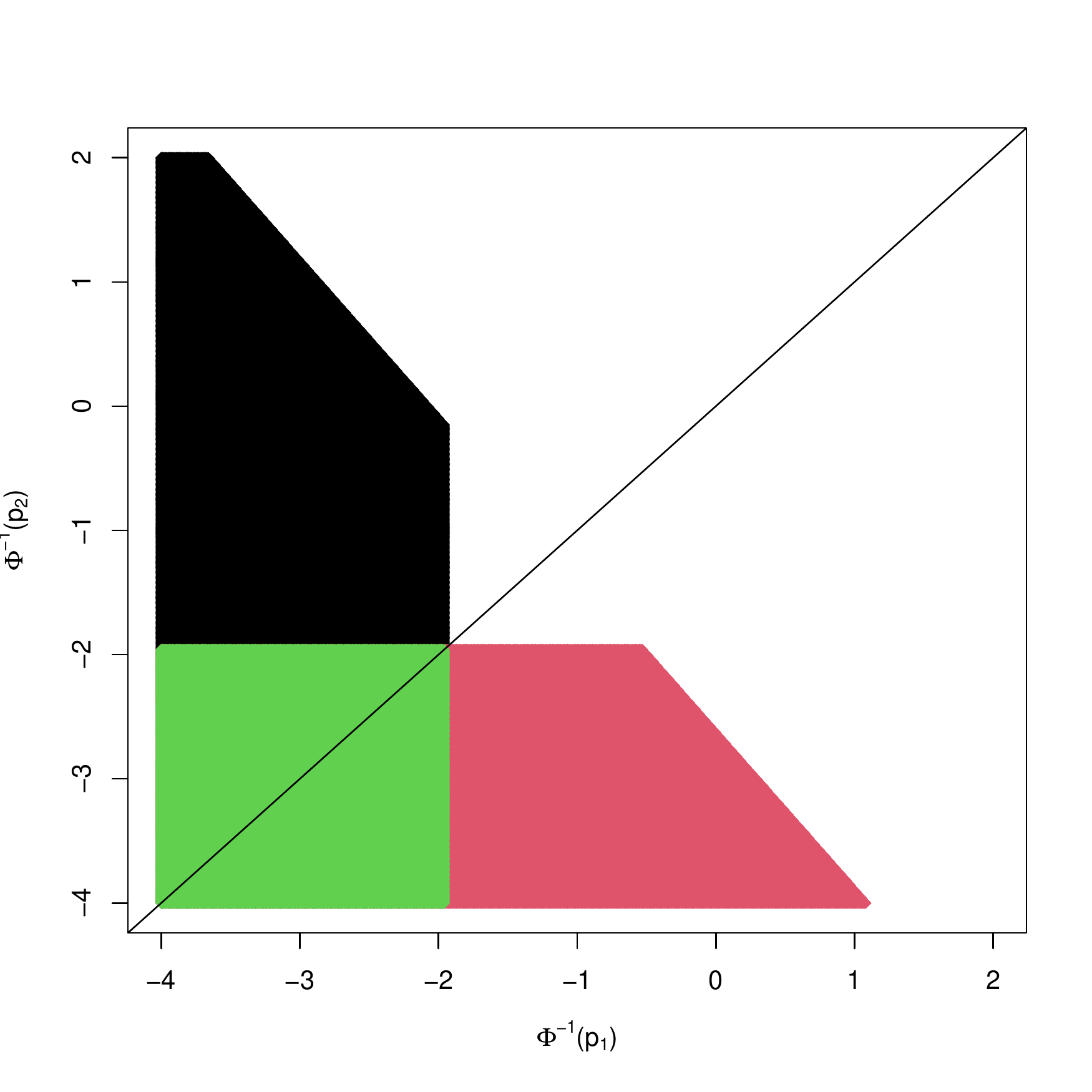} &   
       \includegraphics[width=4.5cm,height=4.5cm]{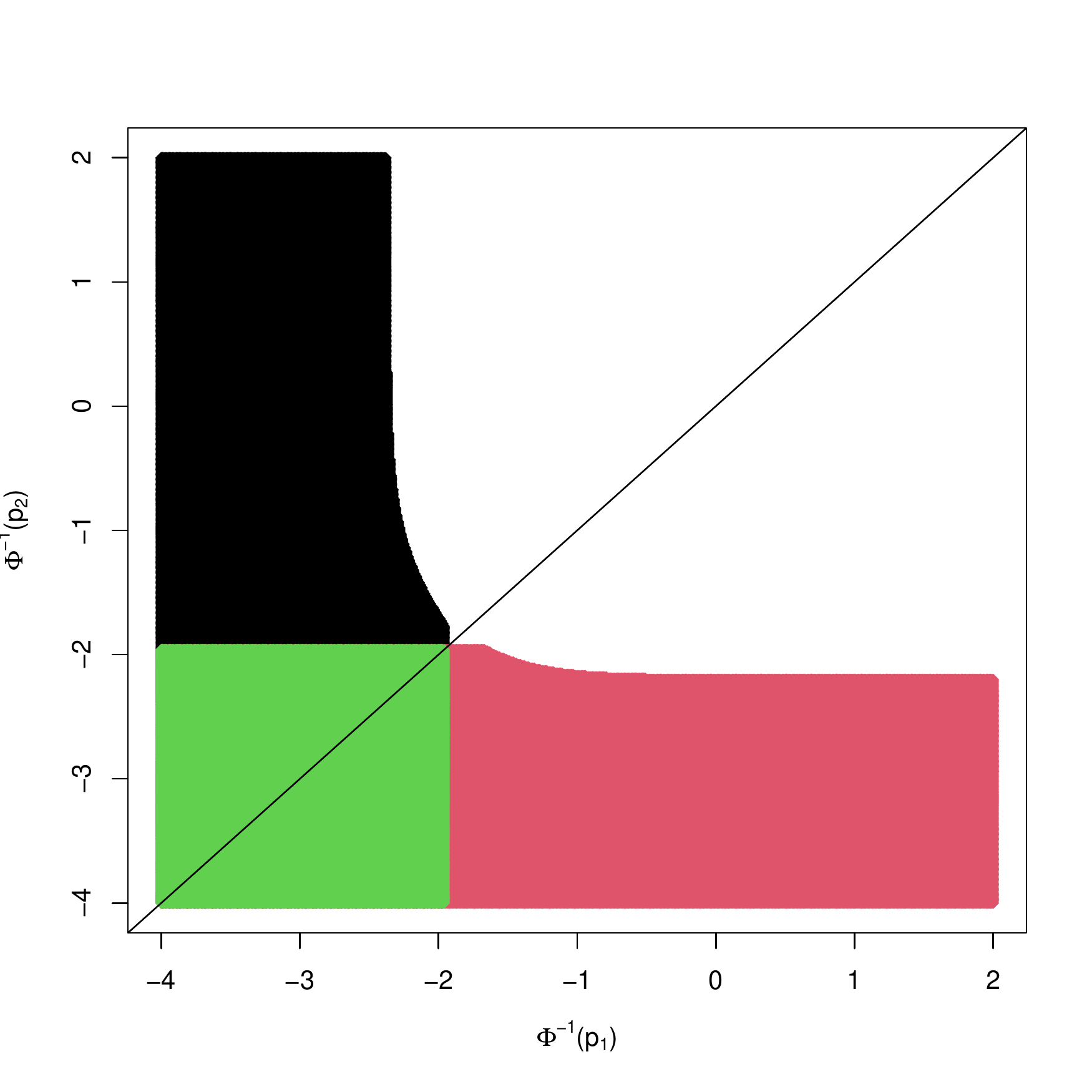} \\ 
  \end{tabular}
 \caption{\label{fig-rejectionregion2}  Rejection regions for the OMT procedures at FWER level of $0.025$, from left to right: the procedure for $\Pi_1$; the procedure for $\Pi_{any}$, which coincides with that of $\Pi_{avg}$;  the procedure for $1/3\times \Pi_{any} + 2/3 \times \Pi_1$.    In green: reject both hypotheses; in red: reject only the second hypothesis; in black: reject only the first hypothesis. The regions are depicted for the $z$-scores, i.e., for $(\Phi^{-1}(p_1), \Phi^{-1}(p_2))$.  }
 
\end{figure}

\section{Discussion}\label{sec-discuss}
We considered the framework of a clearly defined objective, error constraint, and additional desirable restrictions for the design of an analysis of a clinical study. This framework was originally introduced in \cite{Rosset18},  as the extension of the Neyman-Pearson paradigm for $K=1$ to $K>1$: the problem of finding the OMT procedure $\bD$ was cast as an infinite dimensional optimization problem with an appropriately defined objective and constraints. We showed here that 
OMT solutions can help not only in the analysis, but also in the design of clinical trials, where typically a value of the parameter of interest is assumed in order to guide design aspects such as  decisions about sample size allocations.  This approach stands apart from the typical approach when facing multiple hypotheses in clinical trials: rather than  choosing an off-the-shelf procedure that seems appropriate for the problem at hand, we advocate choosing an appropriate objective and the desired constraints, and then seeking the procedure that optimally solves the resulting problem.   

We showed that for $K=2$, this framework can result in an OMT procedure that coincides with an existing one in some cases, and it can result in novel procedures in other cases. Finding OMT solutions can in some cases  help justify the choice of an off-the-shelf procedure, as in our APEX trial example, where Hommel's procedure turned out to be the preferred off-the-shelf procedure in terms of power for the two hypotheses with differing sample sizes considered, though it is still inferior to the novel procedure that is optimal for a relevant objective. In other cases,  using the OMT procedure rather than an existing procedure can result in reduced costs,  since fewer people may need to be recruited in order to achieve the same power (level of the objective).

The setting with $K=2$ hypotheses highlighted the complexities that arise when moving from $K=1$ to $K>1$:  less clear cut definitions  of the elements of the optimization problem, and greater computational difficulty in finding an OMT solution. 
After imposing additional restrictions that result in rejecting only $p$-values at most $\alpha$, the solution turned out to be computationally easy. The OMT procedure for problems \eqref{opt4-nominal}-\eqref{opt4} is obtained by a simple algorithm, and it  provides the same strong FWER control guarantee as existing procedures. We hope that this algorithm,  together with  carefully formulated objectives,  will be useful for researchers designing clinical trials with two endpoints.
Our framework is not restricted to independent test statistics. 
In the fairly common clinical setting of multiple treated groups compared with a single control group, our framework can lead to useful procedures and design decisions (they can be assessed in comparison with Hommel's procedure, which is valid for positive dependence, \citealt{Sarkar97}). Moreover, the framework can be applied in settings other than clinical trials, where the maximin formulation described at the end of \S~\ref{subsec-objective} may be the most useful, since it is not limited  to a single parameter configuration in the objective but rather provides a range of alternatives of interest.

While the connections to our previous work in \cite{Rosset18} are extensive, it is important to highlight the novel theoretical and methodological contributions of the current paper presented in~\S \ref{sec-K2} . Most importantly, we  relax the exchangeability requirement in \cite{Rosset18} and derive computationally efficient and conceptually simple algorithms. Both of these developments are enabled by our adoption of the weak monotonicity  or  marginally nominal-$\alpha$ requirement. 

In this work we only considered a few objectives and constraints out of many that could be interesting and useful for practitioners. For example, in some settings it may be preferable to impose constraints on the tail probability of the false discovery proportion or on the expected number of false discoveries rather than on the FWER. 


A natural follow-up is to extend the solution to $K>2$, as well as to more complex designs with primary and secondary hypotheses. 
After specifying the objectives and constraints, finding the optimal test for $K>2$ may not be trivial. 
Nevertheless, one can evaluate off-the-shelf multiple testing procedures to ensure that they satisfy the requisite constraints and use the desired objective to choose from among them. This is in line with 
a major thrust of this paper, which is to emphasize that the choice of a MTP, even from a set of existing  off-the-shelf procedures, is guided by the choices one makes about the objective of interest and constraints that are imposed. 

\section*{Acknowledgements}
This research was supported by Israeli Science Foundation grant 2180/20.


\appendix

\section{The original design of the APEX trail}

The original trial design \citep{Cohen16-short} was to test the following populations in order: the subgroup with elevated D-dimer level (cohort 1); cohort 1 plus those who were at least 75 years old (cohort 2); all comers (cohort 3). Testing in order  means that if the null hypothesis in a cohort is not rejected at the predefined level $\alpha$, then testing stops and no further discoveries are made. In this trial, the two-sided $p$-value for cohort 1 was above 0.05 and therefore  no efficacy claim can be  made for Betrixapan in any of the patient populations. The analysis in cohorts 2 and 3 can only be considered exploratory despite the fact that the two-sided $p$-values for these cohorts were below 0.05.

\section{Proof of Lemma 1}
\begin{proof}
Assume the Lemma does not hold. WLOG  there exists $(c_1,c_2)$ with $c_1>\alpha$ such that $D_1^*(c_1,c_2)=1$. By weak monotonicity this implies that $D_1^*(c,c_2)=1$ for all $c \le c_1$. Assume that the first null hypothesis is true and the second is false. Let $\epsilon=\frac{c_1-\alpha}{2}$. By our assumption, there exists a $\vartheta_2$ such that $P_{\vartheta_2}(p_2>c_2)\le\epsilon$ and $P_{\vartheta_2}(p_1\leq c_1)=c_1$. Hence
$$
c_1=P_{\vartheta_2}(p_1 \le c_1)=P_{\vartheta_2}(p_1 \le c_1, p_2 \le c_2)+P_{\vartheta_2}(p_1 \le c_1, p_2 > c_2)\le P_{\vartheta_2}(p_1 \le c_1, p_2 \le c_2)+\epsilon.
$$
This implies that $P_{\vartheta_2}(p_1 \le c_1, p_2 \le c_2) \ge c_1-\epsilon= \frac{c_1+\alpha}{2}>\alpha$ which violates strong FWER.
\end{proof}

\section{The optimal split for $\Pi_{any}$ when testing two independent normal means }

\begin{prop}
 For $K=2$ identical one-sided independent normal means problems, with $\btheta_A = (\theta, \theta)$,  and  $\Pi_{\btheta_A} = \Pi_{any, \btheta_A}$, suppose we can allocate a total sample size of $N$ independent observations. Then, the optimal power among all procedures that satisfy strong FWER control and are marginally level $\alpha$,  is achieved when $r$ is $0$ or $1$, for any fixed total sample size $N$.  
\end{prop}
\begin{proof}
Let $rN$ and $(1-r)N$ be the sample sizes, and $\overline X_{1,rN}$ and $\overline X_{2,(1-r)N}$ the sample means for the two hypotheses (for simplicity, assume the sample variance is one). Then
\begin{eqnarray}
&& \hspace{-1cm} \Pi_{any} = \mP_{(\theta,\theta)}(\max(D_1,D_2)=1) \leq \mP_{(\theta, \theta)}\left(\sqrt {N}(r\overline{X}_{1,rN}+(1-r) \overline{X}_{2,(1-r)N})\leq \Phi^{-1}(\alpha)\right).\label{eq-optimalsamplesize-piany}
\end{eqnarray}
The  inequality follows since $\sqrt {N}(r\overline{X}_{1,rN}+(1-r) \overline{X}_{2,(1-r)N})\leq \Phi^{-1}(\alpha)$ is the most powerful level $\alpha$ test for the global null. The distribution of $\sqrt {N}(r\overline{X}_{1,rN}+(1-r) \overline{X}_{2,(1-r)N})$ is normal with mean  $\sqrt{N}\theta$ and variance one. Hence the RHS equals
$$ \mP_{(\theta, \theta)}\left(\sqrt{N}\overline{X}_{1,N}\leq \Phi^{-1}(\alpha)\right) = \mP_{(\theta, \theta)}\left(\sqrt{N} \overline{X}_{2,N}\leq \Phi^{-1}(\alpha)\right),  $$
which is the OMT procedure for $r$ equals 0 or 1, thus completing the proof. (The inequality becomes an equality for these two values only, since for $r\in (0,1)$ the marginally level $\alpha$ requirement leads to a strict inequality in \eqref{eq-optimalsamplesize-piany}.) 
\end{proof}




\label{lastpage}

\end{document}